\documentclass[10pt, a4paper,aps,prl,showpacs,twocolumn,superscriptaddress,nofootinbib]{revtex4-2}   %APS document class style  (PRA,PRL,ETC)
   
\usepackage[utf8]{inputenc}  %Input what you want e.g., é, ł, a, ü
\usepackage[T1]{fontenc}     %Output what you want e.g., é, ł, a, ü
\usepackage[british]{babel}  %Do hyphenation according to british english
\usepackage{lmodern}  %Joe's font
\usepackage[scaled=1.03]{inconsolata} %Monospace font
\usepackage[usenames,dvipsnames]{color} %More color choices
\usepackage[colorlinks,citecolor=green,linkcolor=magenta,urlcolor=magenta]{hyperref}  %Hyperlinks (pink, green, blue)
\usepackage{graphicx} % Package to insert external figures
\usepackage{tikz}
\usepackage[babel]{microtype}  %Improves text justification
\usepackage{amsmath,amssymb,amsthm,bm,mathtools,amsfonts,mathrsfs,bbm,dsfont} %Useful math packages
\usepackage{xspace}  %Useful to add space in macros
\usepackage{multirow}
\usepackage{verbatim}
\usepackage[capitalize,nameinlink]{cleveref}
\usepackage{physics}
\newcommand{\id}{\ensuremath{\mathds{1}}}
\usepackage{bbold}
\usepackage{tcolorbox}

\usepackage[nolist,withpage]{acronym}
\makeatletter
\AtBeginDocument{%
  \renewcommand*{\AC@hyperlink}[2]{%
    \begingroup
      \hypersetup{hidelinks}%
      \hyperlink{#1}{#2}%
    \endgroup
  }%
}

\newtheoremstyle{mystyle}% name
  {6pt}%Space above
  {6pt}%Space below
  {\normalfont}%Body font
  {0pt}%Indent amount
  {\bf}% Theorem head font
  {.}%Punctuation after theorem head
  { }%Space after theorem head 2
  {}%Theorem head spec (can be left empty, meaning ‘normal’)

\theoremstyle{mystyle}
\newtheorem{theorem}{Theorem}

\newtheorem{lemma}{Lemma}

\usepackage{blindtext}

\usepackage{xcolor}
\colorlet{myPurple}{blue!40!red}
\colorlet{myCyan}{cyan!50!gray}

\definecolor{quantumviolet}{HTML}{53257F} %Quantum violet
\definecolor{quantumgray}{HTML}{555555} %Quantum gray
\definecolor{mygray}{gray}{0.95} %Quantum gray

\newtcolorbox[auto counter,number within=section]{boxfigure}[2][]{%
colback=mygray,colframe=myPurple,fonttitle=\bfseries,width=\columnwidth,float*=ht,lower separated=false, halign=justify,title=Box~\thetcbcounter: #2,#1}

\begin{document}
\title{Can every set of incompatible measurements lead to genuine multipartite steering?}

\author{Lucas E. A. Porto}
\email{lucas.porto@lip6.fr}
\affiliation{{Sorbonne Universit\'{e}, CNRS, LIP6, F-75005 Paris, France}}
\author{Lucas Tendick}
\affiliation{Inria, CPHT, LIX, CNRS, École polytechnique, Institut Polytechnique de Paris, Palaiseau, France}
\author{Daniel Cavalcanti}
\affiliation{Algorithmiq Ltd, Kanavakatu 3C 00160 Helsinki, Finland}
\author{Roope Uola}
\affiliation{Department of Physics and Astronomy, Uppsala University, Box 516, 751 20 Uppsala, Sweden}
\affiliation{Nordita, KTH Royal Institute of Technology and Stockholm University, Hannes Alfvéns väg 12, 10691 Stockholm, Sweden}
\author{Marco Túlio Quintino}
\affiliation{{Sorbonne Universit\'{e}, CNRS, LIP6, F-75005 Paris, France}}
\date{\today}  %Date today

\begin{abstract}
Measurement incompatibility and bipartite quantum steering are known to display a strong connection: a set of measurements is incompatible if and only if it can lead to bipartite steering.
Despite such a close link between these concepts in bipartite scenarios, little is known in the multipartite setting, where notions of \textit{genuine} multipartite correlations play major roles.
In this work we prove that, as in the bipartite case, incompatibility is also necessary and sufficient for genuine multipartite steering in any multipartite scenario with a single uncharacterised party.
Interestingly, genuine multipartite steering can be extracted from any set of incompatible measurements using states which are not SLOCC equivalent, such as GHZ and W states.
In contrast, we prove that this result does not hold in scenarios with more than one uncharacterised party, by presenting a set of incompatible measurements that can never lead to genuine multipartite steering in these cases. In order to obtain our main results, we introduce methods tailored for multipartite correlations, paving the way to understanding the role of measurement incompatibility beyond bipartite scenarios. 
\end{abstract}
\maketitle

\begin{figure}
    \centering
    \includegraphics[width=1\linewidth]{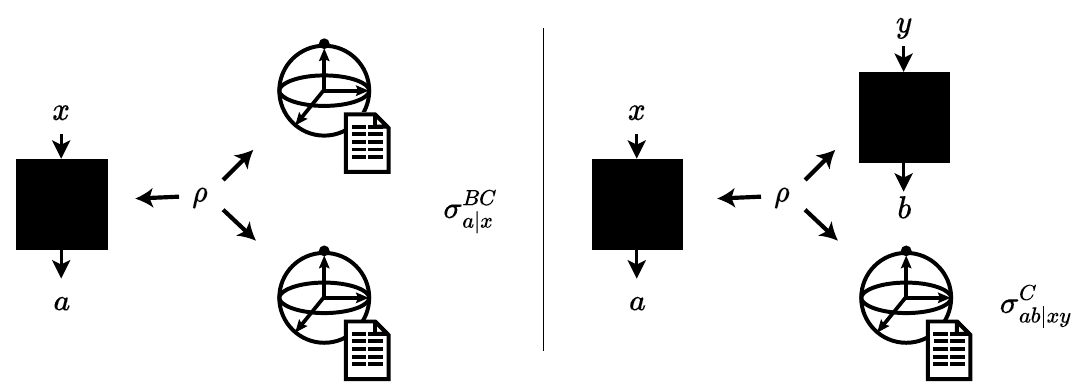}
    \caption{Scheme of tripartite steering scenarios where black boxes represent uncharacterised measurements. 
    In scenarios with a single uncharacterised party (left), we show that any set of incompatible measurements enables genuine multipartite steering. 
    In contrast, in scenarios where more than one party performs uncharacterised measurements (right), we provide sets of incompatible measurements which never lead to genuine multipartite steering.}
    \label{fig:placeholder}
\end{figure}

\indent \textit{Introduction.}\textemdash 
The fact that local measurements performed on distant quantum systems can lead to correlations that resist any classical explanation constitutes one of the most profound insights of quantum theory. Bell's 1964 seminal work~\cite{Bell1964seminal} can be understood as an answer to the Gedankenexperiment proposed by Einstein, Podolsky and Rosen (EPR)~\cite{EPR1935paper}, in which the notions of entanglement~\cite{Horodecki2009Entanglement} and measurement incompatibility ~\cite{Guehne2023IncompReview, Heinosaari2016Invitation} (i.e., the impossibility of measuring certain observable quantities simultaneously) are used to argue against the completeness of quantum theory.
Bell's work formalises the classical reasoning of EPR in a theory-independent way, and points to an actual experiment that can decide whether nature is in agreement with local hidden variable theories.

The fact that Bell's approach is agnostic to the underlying theory governing the studied systems is arguably its most disruptive feature.
It enables rather deep insights about the foundations of nature itself~\cite{Brunner2014BellReview,Goldstein2011Bell} and, furthermore, it opens possibilities for fundamentally novel kinds of information processing protocols, with minimal assumptions on the inner workings of devices~\cite{Ekert1991QKD,Colbeck2009PhDRandom,Colbeck2011Random, Pironio2010Random,Acin2007QKD,Arslan2025DIQKD,Supic2020selftesting}.
It allows, for example, to certify the entanglement of quantum states solely by analysing the statistics produced by performing uncharacterised measurements on them~\cite{Moroder2013device,Guhne2009EntanglementDetection}.

The advantages of such an approach, however, come with some costs, which are mainly associated with the difficulty of achieving Bell nonlocal correlations within quantum theory.
The main resources necessary for doing so are entanglement and measurement incompatibility, although they are known not to be sufficient. 
There are entangled states which can never violate a Bell inequality \cite{Werner1989LHV, Barett2002LHV} and, similarly, there are incompatible measurements that are unable to produce Bell nonlocal correlations~\cite{Quintino2016incompatible,Hirch2017IncompBell, Bene2018IncompBell}.
From this perspective, it is natural to investigate intermediate scenarios, where some partial assumptions are made about the underlying physics of the experiment.
This is one of the motivations for the study of the so-called quantum steering~\cite{Wiseman2007Steering,Cavalcanti2017Review,Uola2020ReviewSteering}, which takes place in scenarios where (at least) one of the parties has a fully characterised quantum system, while the others may have untrusted devices.

As in the case of Bell nonlocality, entanglement and measurement incompatibility are the key resources for producing steering, although the requirements are somewhat weaker. Entanglement is not sufficient for steering, and steering is not sufficient for Bell nonlocality~\cite{Werner1989LHV, Wiseman2007Steering,Qunitino2015Inequivalence}.
However,  measurement incompatibility is not only a necessary condition for steering, but also a sufficient one: any set of incompatible measurements can be used to demonstrate steering \cite{Uola2014JM, Quintino2014JM,Uola2015Mapping}.
So far, this close relation between incompatibility and steering has been exclusively investigated in steering scenarios with only two parties.
In multipartite scenarios, the connection between these two concepts remains largely unexplored, especially when the notion of \textit{genuine} multipartite steering is considered \cite{Cavalcanti2015GMS, Cavalcanti2017Review}.

In this paper, we prove that in multipartite scenarios with a single uncharacterised party, every set of incompatible measurements can lead to genuine multipartite steering. Moreover, 
there are states which are not SLOCC equivalent (i.e., that cannot be converted into one another using stochastic local operations and classical communication~\cite{Dur2000Wstates}), such as GHZ and W states, that  extract genuine multipartite steering from \textit{any} set of incompatible measurements.
This contrasts with the bipartite case, where only pure states with full Schmidt rank\footnote{In the bipartite case, pure states are SLOCC equivalent if and only if they have the same Schmidt rank ~\cite{Horodecki2009Entanglement}.} can demonstrate steering of any set of incompatible measurements~\cite{Quintino2014JM, Uola2014JM, Uola2015Mapping}. 

We then show that incompatibility is insufficient for genuine multipartite steering in scenarios with two or more uncharacterised parties. 
This is shown by identifying a necessary and sufficient condition for sets of measurements to enable genuine multipartite steering in these scenarios, and by proving that this condition is not implied by the incompatibility of the individual sets of measurements.
This condition suggests a generalised form of measurement incompatibility, applicable to local measurements on multipartite systems and which might be of independent interest, beyond the scope of this work.

\textit{Bipartite steering and measurement incompatibility.}\textemdash Let us start by reviewing the scenarios and concepts that motivate our discussion.
Consider a bipartite quantum state $\rho^{AB}$, shared by two spatially separated observers, Alice and Bob.
Let $ \{A_{a \vert x}\}_{a, x}$ be a set of positive operator-valued measure (POVM) measurements that Alice can perform on her share of the system, where $A_{a \vert x}$ is the POVM element associated with the outcome $a$ of the measurement setting $x$.
Then, when Alice measures $x$ and obtains the outcome $a$, Bob's unnormalised state can be described by 
\begin{equation}\label{eq:bipartite_assemblage}
    \sigma_{a \vert x}^B = \Tr_A\left(A_{a \vert x} \otimes \id \rho^{AB}\right).
\end{equation}
The set of such states $\{\sigma^B_{a \vert x}\}_{a, x}$ is called a \textit{steering assemblage}.

We say that a steering assemblage is \textit{unsteerable} when it does not exhibit steering, that is, when it can be decomposed as
\begin{equation}\label{eq:unsteerable_assemblage}
    \sigma^B_{a \vert x} = \sum_\lambda p(\lambda) p(a \vert x, \lambda) \rho^B_\lambda,
\end{equation}
where $p(\lambda)$ and $p(a \vert x, \lambda)$ are probabilities, and  $\rho^B_\lambda$ are quantum states. 
An unsteerable assemblage can be understood in the following way: Bob holds a quantum state $ \rho^B_{\lambda} $ with probability $p(\lambda)$ and upon learning that Alice measured $x$ and obtained outcome $a$, Bob's knowledge is updated according to $p(a \vert x, \lambda)$ \cite{Uola2020ReviewSteering}. 

Notice that if Alice and Bob share a separable state $ \rho^{AB} = \sum_{\lambda} p({\lambda}) \rho^A_{\lambda} \otimes \rho^B_{\lambda} $, Bob's steering assemblage can always be decomposed as in \cref{eq:unsteerable_assemblage}, i.e., it does not exhibit steering.
Therefore, if Bob holds an assemblage that demonstrates steering, he knows that the state shared with Alice must have been entangled.

Importantly, this conclusion does not require the precise knowledge of the measurements performed by Alice.
In this sense, steering allows for the certification of entanglement in scenarios where the measurements performed by one party are not characterised, that is, in a \textit{one-sided} device-independent manner.
However, not all entangled states are able to produce steering in this standard scenario \cite{Werner1989LHV, Wiseman2007Steering,Qunitino2015Inequivalence}.

Another necessary ingredient for steering is the \textit{incompatibility} of the measurements performed by Alice.
A set of measurements $ \{A_{a \vert x}\}_{a, x}$ is \textit{jointly measurable}, or \textit{compatible}, if there exists a single measurement $ \lbrace G_{\lambda} \rbrace_{\lambda} $ out of which the statistics of each measurement $x$ can be recovered via classical post-processing, i.e.,
\begin{equation}\label{Eq:DefJM}
    A_{a \vert x} = \sum_{\lambda} p(a \vert x, \lambda) G_{\lambda}.
\end{equation}
Otherwise, the measurements $\{A_{a \vert x}\}_{a,x}$ are said to be \textit{incompatible}.
If the measurements performed by Alice in a steering test are compatible, the assemblage given by \cref{eq:bipartite_assemblage} is necessarily unsteerable.

Conversely, and unlike entanglement, incompatibility is also sufficient for steering.
That is, every set of incompatible measurements leads to quantum steering when applied to an appropriate entangled state \cite{Uola2014JM, Quintino2014JM}.
In particular, the $d-$dimensional maximally entangled state $\ket{\Phi^{+}} := \tfrac{1}{\sqrt{d}} \sum_{i=0}^{d-1} \lvert ii \rangle $ can reveal the incompatibility of any set of incompatible measurements, as it holds that 
\begin{align}
\sigma_{a \vert x}^B = \mathrm{Tr}_A \Big[ (A_{a \vert x} \otimes \mathds{1}^B)  \lvert \Phi^{+} \rangle \langle \Phi^{+} \rvert \Big] = \dfrac{A^{T}_{a \vert x}}{d},  \label{Eq: MaxEntangledAssemblage}  
\end{align}
where the transpose is taken with respect to the Schmidt basis of $ \ket{\Phi^{+}}$. 
The assemblage in \cref{Eq: MaxEntangledAssemblage} is steerable if and only if the measurement assemblage $\{A_{a \vert x}\}_{a,x}$ is incompatible. This establishes a close connection between quantum steering and measurement incompatibility in bipartite scenarios.
The goal of this work is to investigate whether such a strong relation between these two concepts still holds in multipartite steering scenarios.

\textit{Genuine multipartite steering.}\textemdash As we mentioned above, bipartite steering can be understood as a one-sided device-independent certification of entanglement.
In a similar fashion, multipartite steering is also defined aiming at the certification of multipartite entanglement.
However, in multipartite scenarios more detailed entanglement structures arise \cite{Horodecki2009Entanglement}, which subsequently lead to richer multipartite steering structures \cite{Cavalcanti2017Review}.

The notion of multipartite entanglement that captures the idea of entanglement distributed amongst all parties is the so-called \textit{genuine multipartite entanglement} (GME) \cite{Guhne2009EntanglementDetection}.
A state is said to exhibit genuine multipartite entanglement when it is not \textit{bi-separable}, i.e., when it cannot be decomposed as a convex combination of states which are separable according to bipartitions.
For example, in the tripartite case a bi-separable state $\rho_{\text{Bi-Sep}}^{ABC}$ can always be decomposed as
\begin{align}\label{eq:bi-sep_states}
    \rho_{\text{Bi-Sep}}^{ABC} &= \sum_\lambda p(\lambda) \rho^A_\lambda \otimes \rho^{BC}_\lambda \nonumber\ + \sum_\mu p(\mu) \rho^{B}_\mu \otimes \rho^{AC}_\mu \\
    &+ \sum_\nu p(\nu) \rho^C_\nu \otimes \rho^{AB}_\nu,
\end{align}
with $\sum_\lambda p(\lambda) + \sum_\mu p(\mu) + \sum_\nu p(\nu) = 1$.
The definition of multipartite steering we consider in this work will be thus tailored to certify GME, and for that reason it is termed \textit{genuine multipartite steering} (GMS) \cite{Cavalcanti2017Review, Cavalcanti2015GMS} (see \cite{He2013GMS} for a non-equivalent definition of GMS that captures other aspects of multipartite correlations).

Another important detail to properly define multipartite steering scenarios is that one needs to specify not only the number of parties, but also which parties should be treated as uncharacterised.
For instance, in a tripartite scenario we can consider both the case where a single party performs uncharacterised measurements and steers the remaining two, or the case where two parties perform uncharacterised measurements and steer the remaining one. 

In the former case, of \textit{one-sided} device-independent scenarios, and assuming that Alice is steering Bob and Charlie by performing measurements $\{A_{a \vert x}\}_{a, x}$ on her share of a tripartite state $\rho^{ABC}$, the steering assemblage $\{\sigma_{a \vert x}^{BC}\}_{a, x}$ is given by
\begin{equation}\label{eq:assemblage_one-sided}
    \sigma^{BC}_{a \vert x} = \Tr_A\Big[(A_{a \vert x} \otimes \id^{BC}) \rho^{ABC}\Big].
\end{equation}
If $\rho^{ABC}$ is bi-separable, such an assemblage can always be decomposed as
\begin{align}\label{eq:tripartite_one-sided_bisep_assemblage}
\sigma_{a|x}^{BC}= &\sum_\lambda p(\lambda) p(a|x\lambda)\rho_\lambda^{BC}\nonumber +\sum_\mu p(\mu) \sigma_{a|x\mu}^B\otimes\rho_\mu^C\\
    +&\sum_\nu p(\nu)\rho_\nu^B\otimes\sigma_{a|x\nu}^C,
\end{align}
where $\sum_a \sigma^B_{a \vert x \mu} = \rho^B_\mu$, $\sum_a \sigma_{a \vert x \nu}^C = \rho_\nu^C$ and $\sum_\lambda p(\lambda) + \sum_\mu p(\mu) + \sum_\nu p(\nu) = 1$.
On the right-hand side of \cref{eq:tripartite_one-sided_bisep_assemblage}, the first term corresponds to the situation where Alice cannot steer the system $BC$, but Bob and Charlie can share arbitrary entanglement between each other.
The second (third) term is separable and represents the case where Alice can only steer Bob (Charlie) \cite{Cavalcanti2015GMS, Cavalcanti2017Review}.

In the spirit of defining genuine multipartite steering as a way to certify genuine multipartite entanglement, we say that an assemblage $\{\sigma_{a \vert x}^{BC}\}_{a, x}$ exhibits GMS if it cannot be decomposed as \cref{eq:tripartite_one-sided_bisep_assemblage}.
In this sense, for an assemblage to demonstrate GMS, it is necessary that the underlying state $\rho^{ABC}$ exhibits GME.
However, it is also known that there exist genuine multipartite entangled states that cannot generate any form of multipartite steering~\cite{Bowles2016GME}.

In a similar fashion, in a \textit{two-sided} device-independent scenario, where two parties, say Alice and Bob, steer a third party, Charlie, by performing local measurements $\{A_{a \vert x}\}_{a, x}$ and $\{B_{a \vert x}\}_{b, y}$ on their respective shares of a tripartite state $\rho^{ABC}$, the steering assemblage $ \{ \sigma_{ab \vert xy}^{C}\}_{a,b,x,y}$ has the form
\begin{equation}\label{eq:assemblage_two-sided}
    \sigma^{C}_{ab \vert xy} = \Tr_{AB}\Big[(A_{a \vert x} \otimes B_{b \vert y} \otimes \id^{C}) \rho^{ABC}\Big].
\end{equation}
In this case, assemblages generated by bi-separable states 
can be decomposed as
\begin{align}\label{eq:two-sided_bisep_assemblage}
    \sigma_{ab|xy}^{C}= &\sum_\lambda p(\lambda) p(a|x\lambda)\sigma_{b \vert y \lambda}^{C}\nonumber
    +\sum_\mu p(\mu) p(b|y\mu)\sigma_{a \vert x \mu}^{C}\\
    +&\sum_\nu p(\nu)p(ab|xy\nu)\rho_\nu^C,
\end{align}
where $\sum_b \sigma_{b \vert y \lambda}^C = \rho^C_\lambda$, $\sum_a \sigma_{a \vert x \mu}^C = \rho^C_\mu$, $p(ab|xy\nu)$ is an arbitrary quantum correlation, and $\sum_\lambda p(\lambda) + \sum_\mu p(\mu) + \sum_\nu p(\nu) = 1$.
In \cref{eq:two-sided_bisep_assemblage}, the first (second) term on the right-hand side corresponds to the case where only Alice (Bob) can steer Charlie, and the last term to the case where Alice and Bob cannot jointly steer Charlie.
Notice that in this last term, Alice and Bob may share arbitrary quantum nonlocal correlations \cite{Cavalcanti2015GMS, Cavalcanti2017Review}.

If an assemblage $\{ \sigma_{ab \vert xy}^{C}\}_{a,b,x,y}$ does not admit a decomposition of the form \cref{eq:two-sided_bisep_assemblage}, it exhibits GMS.
Also in this case, GME is a necessary resource for GMS, but there are GME states that are not able to produce GMS in these scenarios \cite{Bowles2016GME}.

Besides GME, for assemblages of the form \cref{eq:assemblage_one-sided} and \cref{eq:assemblage_two-sided} to demonstrate genuine multipartite steering, it is necessary that the measurement sets $\{A_{a \vert x}\}_{a, x}$ and $\{B_{b \vert y}\}_{b, y}$ be incompatible.
In the bipartite case, the corresponding condition of measurement incompatibility turns out to be sufficient for steering, as we have previously discussed \cite{Uola2014JM, Quintino2014JM}.
In this work, we investigate whether this remains valid in the multipartite case, i.e., whether for any set of incompatible measurements $\{A_{a \vert x}\}_{a, x}$ and $\{B_{b \vert y}\}_{b, y}$, we can find states $\rho^{ABC}$ for which the assemblages \cref{eq:assemblage_one-sided} and \cref{eq:assemblage_two-sided} demonstrate GMS.

Although above we only provided the precise definitions for tripartite scenarios, in the following we state our results for scenarios with an arbitrary number of parties.
For the associated definitions, as well as the corresponding proofs, we refer to the End Matter.

\textit{One-sided device-independent scenarios.}\textemdash 
In multipartite scenarios with one untrusted party, we show that, as it is the case for bipartite scenarios, incompatibility is sufficient for genuine multipartite steering.
That is, any set of incompatible measurements is able to generate assemblages that demonstrate genuine multipartite steering.
\begin{theorem}\label{thm:one-sided_incomp_implies_gms}
    Let $\ket{\psi^{(1 \ldots n)}} \in \mathbb{C}^{d_1} \otimes \mathbb{C}^{d^{\otimes (n-1)}}$ be an $n$-partite state whose Schmidt decomposition in the bipartition $1|2\ldots n$ has the form $\ket{\psi} = \sum_{i = 0}^{d_1 - 1}\sqrt{p_i} \ket{\phi_i^{(1)}} \otimes \ket{\zeta_i^{(2\ldots n)}}$, where $p_i > 0$ and $\ket{\zeta_i}^{2 \ldots n}$ is invariant under any permutation of the parties $2 \ldots n$, for all $i \in \{0, \ldots, d_1 -1\}$.
    
    A set of  $d_1$-dimensional measurements $ \{A_{a \vert x}\}_{a,x}$ is incompatible if and only if the assemblage $\sigma_{a|x}^{2 \ldots n} = \Tr_1(A_{a|x}\otimes \id^{2 \ldots n} \ketbra{\psi^{(1 \ldots n)}})$ exhibits genuine multipartite steering.
\end{theorem}
\begin{proof}
    Since compatible measurements are not able to produce steering, we just need to show that the incompatibility of $\{A_{a \vert x}\}_{a,x}$ is also sufficient for $\{\sigma_{a|x}^{BC}\}_{a, x}$ to exhibit GMS. 
    We do so by proving its contrapositive, if $\{\sigma_{a|x}^{BC}\}_{a, x}$ is not GMS, then $\{A_{a \vert x}\}_{a,x}$ is jointly measurable.  
    
    For the tripartite case, this can be done as follows.
    Assume that the assemblage $\sigma_{a|x}^{BC} = \Tr_A(A_{a|x}\otimes \id^{BC} \ketbra{\psi^{ABC}})$ admits a decomposition as in \cref{eq:tripartite_one-sided_bisep_assemblage}.
    Regardless of the measurement elements $A_{a \vert x}$, since the state $\ket{\psi^{ABC}}$ is invariant under the permutation of $B$ and $C$, it follows that every operator on the right-hand side of \cref{eq:tripartite_one-sided_bisep_assemblage} must also be invariant under such permutation (see \cref{lem:psd_orthogonality_decomposition} and \cref{lem:flip_invariance} in the End Matter). 
    This implies that the terms of the form $\sigma_{a|x\mu}^B\otimes\rho_\mu^C$ in \cref{eq:tripartite_one-sided_bisep_assemblage} are such that $\sigma_{a|x\mu}^B = \Tr(\sigma_{a|x\mu}^B) \rho_\mu^C$, and similarly for the terms $\rho_\nu^B\otimes\sigma_{a|x\nu}^C$.
    Therefore, the assemblage $\{\sigma_{a|x}^{BC}\}_{a, x}$ does not demonstrate steering at all, i.e., $\sigma_{a \vert x}^{BC} = \sum_{\lambda'} p(\lambda') p(a|x\lambda')\rho_{\lambda'}^{BC}$.
    
    Now, since the Schmidt rank of $\ket{\psi^{ABC}}$ in the bipartition $A|BC$ is $d_1$, we can use this decomposition of the assemblage to construct a parent measurement for $\{A_{a \vert x}\}_{a,x}$ (see \cref{lem:pos_left_inverse} in the End Matter).
    This proves that if the set of measurements $\{A_{a \vert x}\}_{a,x}$ is incompatible, then the assemblage $\Tr_A(A_{a|x}\otimes \id^{BC} \ketbra{\psi^{ABC}})$ must exhibit genuine multipartite steering.
    A proof for the general $n$-partite case is discussed in the End Matter.
\end{proof}
\cref{thm:one-sided_incomp_implies_gms} provides a sufficient condition for identifying pure states $\ket{\psi^{(1 \ldots n)}}$ that demonstrate GMS when subjected to any set of incompatible measurements.
This condition encompasses some of the most notorious GME states, such as the GHZ state $\ket{\textup{GHZ}}:=\frac{1}{\sqrt{2}}(\ket{0}^{\otimes n} + \ket{1}^{\otimes n} )$ and some Dicke states~\cite{Gulati2026Dicke, Marconi2026Symm}, including the $n$-qubit W state, which for $n=3$ reads $\ket{W}:=\frac{1}{\sqrt{3}}(\ket{001} + \ket{010} + \ket{100})$.
We then see that sets of states which are not SLOCC equivalent~\cite{Dur2000Wstates,Horodecki2009Entanglement} are equally useful to extract genuine multipartite steering out of incompatible measurements, in contrast with the bipartite case, where only pure full Schmidt rank can demonstrate steering of arbitrary incompatible measurements~\cite{Quintino2014JM, Uola2014JM, Uola2015Mapping}. 

\textit{Multiple-sided device-independent scenarios.}\textemdash 
In contrast with the one-sided device-independent case, in the following we show that incompatibility by itself is not sufficient for genuine multipartite steering in scenarios with more than one untrusted party.
More precisely, in the tripartite case this means that there are sets of measurements $ \lbrace A_{a \vert x}\rbrace_{a,x} $ and $ \lbrace B_{b \vert y}\rbrace_{b,y} $ which are both incompatible, but when applied on any $\rho^{ABC}$, the assemblage given by \cref{eq:assemblage_two-sided} does not exhibit genuine multipartite steering.

\begin{theorem}\label{thm:two_sided_incomp_but_no_gms}
    Let $M_{\pm \vert x} = \frac{1}{2}\left(\id \pm \frac{1}{\sqrt[4]{2}}\sigma_x\right)$ for $x \in \{1, 2\}$ where $\sigma_1$, $\sigma_2$ are the Pauli matrices $X$ and $Y$.
    For any $n$-partite state $\rho^{1 \ldots n}$ and any $k > 1$, the assemblage $\sigma^{k+1 \ldots n}_{a_1 \ldots a_k \vert x_1 \ldots x_k} = \Tr_{1 \ldots k}(M_{a_1 \vert x_1} \otimes \ldots \otimes M_{a_k \vert x_k} \otimes \id^{k+1 \ldots n} \rho^{1 \ldots n})$ does not exhibit genuine multipartite steering.
    Meanwhile, the set of measurements $\{M_{a \vert x}\}_{a, x}$ is known to be incompatible \cite{Busch1986JM, Guehne2023IncompReview}.
\end{theorem}
\begin{proof}
    For the tripartite case with two untrusted parties, notice that if we are able to write
    \begin{equation}\label{eq:bi-partite_jm_no_F}
        M_{a \vert x} \otimes M_{b \vert y} = \sum_{\lambda} p(a \vert x, \lambda) G_{\lambda,b \vert y} + \sum_{\mu} p(b \vert y, \mu) H_{\mu, a \vert x},
    \end{equation}
    where $G_{\lambda, b \vert y} \geq 0$, $H_{\mu, a \vert x} \geq 0$ with $\sum_{b}G_{\lambda, b \vert y} = G_\lambda$ and $\sum_a H_{\mu, a \vert x} = H_\mu$, it follows that for any state $\rho^{ABC}$, the assemblage $\sigma_{a b \vert x y}^C = \Tr_{AB}(M_{a \vert x} \otimes M_{b \vert y} \otimes \id^C \rho^{ABC})$ does not exhibit GMS, as it can be decomposed as a particular case of \cref{eq:two-sided_bisep_assemblage}.
    Therefore, the proof of \cref{thm:two_sided_incomp_but_no_gms} for the tripartite case amounts to finding a decomposition of the form of \cref{eq:bi-partite_jm_no_F} for the noisy Pauli measurements described in the statement of the theorem.
    In the End Matter, we construct explicit sets of $G_{\lambda,b \vert y}$ and $ H_{\mu, a \vert x}$ such that \cref{{eq:bi-partite_jm_no_F}} holds.
\end{proof}
Since unsteerable assemblages cannot lead to nonlocal Bell correlations, \cref{thm:two_sided_incomp_but_no_gms} also ensures that the pair of noisy Pauli measurements from \cref{thm:two_sided_incomp_but_no_gms} also cannot lead do genuine multipartite Bell nonlocality~\cite{Brunner2014BellReview}. Also, since every pair of incompatible dichotomic measurements lead to bipartite Bell nonlocality~\cite{Wolf2009IncompCHSH},  \cref{thm:two_sided_incomp_but_no_gms} provides an example of measurements which can lead to bipartite Bell nonlocality, but cannot lead to genuine multipartite Bell nonlocality.

The proof of \cref{thm:two_sided_incomp_but_no_gms} builds up on the existence of the decomposition \cref{eq:bi-partite_jm_no_F}, which resembles a generalised notion of joint measurability, applicable to local measurements in a bipartite scenario. 
In the first term on the right-hand side of \cref{eq:bi-partite_jm_no_F}, Alice's measurements outcomes come from classical post-processings of a `parent' quantum object $\{G_{\lambda, b \vert y}\}_{\lambda, b, y}$ which does not depend on her measurement setting.
Meanwhile, in the second term it is Bob's outcomes that can be obtained via classical post-processings of a quantum object $\{H_{\mu, a \vert x}\}_{\mu, a, x}$ which does not depend on his choice of measurement.

In fact, a further generalisation of of \cref{eq:bi-partite_jm_no_F} could also have been considered.
Namely, if a pair of measurements $\{A_{a \vert x}\}_{a, x}$ and $\{B_{b \vert y}\}_{b, y}$ can be decomposed as 
\begin{align}\label{eq:bipartite-jm}
    A_{a \vert x} \otimes B_{b \vert y} &= \sum_{\lambda} p(a \vert x, \lambda) G_{\lambda,b \vert y} + \sum_{\mu} p(b \vert y, \mu) H_{\mu, a \vert x} \nonumber \\
    &+ \sum_{\tau} p(ab \vert xy, \tau) F_{\tau},
\end{align}
where $G_{\lambda, b \vert y} \geq 0$, $H_{\mu, a \vert x} \geq 0$, $F_{\tau} \geq 0$, $\sum_{b}G_{\lambda, b \vert y} = G_\lambda$, $\sum_a H_{\mu, a \vert x} = H_\mu$, and $p(ab \vert xy, \tau)$ is a quantum correlation, it follows that such a pair of measurements does not enable GMS in tripartite scenarios.
That is, for any state $\rho^{ABC}$, the assemblage $\Tr_{AB}(A_{a \vert x} \otimes B_{b \vert y} \otimes \id \rho^{ABC})$ can always be written as \cref{eq:two-sided_bisep_assemblage}.

Notably, \cref{eq:bipartite-jm} characterises precisely the pairs of sets of measurements which enable GMS in two-sided device-independent scenarios, according to the following theorem, whose proof is discussed in the End Matter.
\begin{theorem}\label{thm:bi_partite_jm_iff_no_two-sided_gms}
    Two sets of measurements $\lbrace A_{a \vert x}\rbrace_{a,x} $ and $\lbrace B_{b \vert y}\rbrace_{b,y} $ can be decomposed as \cref{eq:bipartite-jm} if and only if for any state $\rho^{1 \ldots n}$ the assemblage $\Tr_{AB}(A_{a \vert x} \otimes B_{b \vert y} \otimes \id^{3 \ldots n} \rho^{1 \ldots n})$ does not exhibit genuine multipartite steering.
\end{theorem}
\cref{thm:bi_partite_jm_iff_no_two-sided_gms} shows that \cref{eq:bipartite-jm} plays an important role in analysing measurement incompatibility in multipartite scenarios. Note that, while checking if two sets of measurements can be decomposed as in \cref{eq:bipartite-jm} may be a computationally hard task,  checking if two sets of measurements can be decomposed as in \cref{eq:bi-partite_jm_no_F} can be done via standard semidefinite programming (SDP) methods~\cite{Cavalcanti2017Review}. For that it is enough to `absorb' the distributions inside the POVM elements, as done in the standard SDP approach to joint measurarability and steering~\cite{Cavalcanti2017Review}.

\textit{Discussion.}\textemdash 
In this work, we showed that the strong relation between steering in bipartite scenarios and measurement incompatibility extends only partially to multipartite scenarios, when the notion of genuine multipartite steering is taken into account.
More specifically, in multipartite scenarios with a single uncharacterised party, we proved that any set of incompatible measurements is able to generate genuine multipartite steering (\cref{thm:one-sided_incomp_implies_gms}).
However, in scenarios with multiple uncharacterised parties, it turns out that incompatibility does not necessarily imply genuine multipartite steering, as \cref{thm:two_sided_incomp_but_no_gms} provides an example of incompatible sets of measurements which do not produce genuine multipartite steering when applied to any quantum state.
This shows that steering in the multipartite case and the standard notion of incompatibility are not always in perfect correspondence.

As part of our methods, \cref{eq:bipartite-jm} presents a definition that may be interpreted as generalisation of joint measurability for spatially distributed measurements and provides a necessary and sufficient condition for sets of measurements to lead to genuine multipartite steering in scenarios with two uncharacterised parties.
While we could not provide a natural operational interpretation for \cref{eq:bipartite-jm}, we believe it may exist. 
One could also consider similar generalisations of joint measurability, for example, asking for normalisation, or for a tensor product structure for the `parent' measurements, which should find operation and practical value beyond the relationship with multipartite steering.
A deeper exploration of these definitions, their properties and operational meaning deserves future work.

\textit{Acknowledgements.}\textemdash We thank Arthur C. R. Dutra, Lina Vandré and Paul Skrzypczyk for fruitful discussions.
This work was funded by QuantEdu France, a State aid managed by the French National Research Agency for France 2030 with the reference ANR-22-CMAS-0001, the Swedish Research Council
(grant no. 2024-05341), the Wallenberg Initiative on Networks and Quantum Information (WINQ), and by the European Union under the Marie Skłodowska-Curie Actions (MSCA) through the QNETS project (grant agreement ID: 101208259). Views and opinions expressed are however those of the author(s) only and do not necessarily reflect those of the European Union or the European Education and Culture Executive Agency (EACEA). Neither the European Union nor EACEA can be held responsible for them. LT acknowledges funding from the ANR through the JCJC grant LINKS (ANR-23-CE47-0003).

\bibliography{bibliography.bib}

\clearpage

\section{End Matter}
\subsection{Definitions for \texorpdfstring{$n$}{n}-partite scenarios}
An $n$-partite quantum state $\rho$ is bi-separable if it admits a decomposition of the form $\rho = \sum_\lambda p(\lambda)\rho^{K^{(\lambda)}_1} \otimes \rho^{K^{(\lambda)}_2}$,
where, for each $\lambda$, $K^{(\lambda)}_1$ and $K^{(\lambda)}_2$ form a bipartition of $\{1. \ldots, n\}$, i.e., $K^{(\lambda)}_1$ and $K^{(\lambda)}_2$ are non-empty disjoint subsets of $\{1, \ldots, n\}$ with $K^{(\lambda)}_1 \cup K^{(\lambda)}_2 = \{1, \ldots, n\}$.
Otherwise, $\rho$ is genuine multipartite entangled.\\
\indent Following the discussion in the main text, in an $n$-partite steering scenario with $k$ untrusted parties, i.e., in a $k$-sided device-independent scenario, an assemblage $ \{\sigma^{k+1 \ldots n}_{a_1 \ldots a_k \vert x_1 \ldots x_k}\}$ does not exhibit genuine multipartite steering if it can be decomposed as
\small
\begin{align}\label{eq:general_bisep_assemblage}
    \sigma^{k+1 \ldots n}_{a_1 \ldots a_k \vert x_1 \ldots x_k} & = \sum_\lambda p(\lambda)p(a_1 \ldots a_k \vert x_1 \ldots x_k, \lambda) \rho^{k+1 \ldots n}_\lambda \nonumber \\
    & + \sum_\mu p(\mu)p(a_{L_1^{(\mu)}} \vert x_{L_1^{(\mu)}}, \mu) \sigma^{k+1 \ldots n}_{a_{L_2^{(\mu)}} \vert x_{L_2^{(\mu)}}, \mu} \nonumber \\
    & + \sum_\nu p(\nu)\sigma^{K_1^{(\nu)}}_{a_1 \ldots a_k \vert x_1 \ldots x_k, \nu} \otimes \rho^{K_2^{(\nu)}}_\nu \nonumber \\
    & + \sum_\omega p(\omega)\sigma^{K_1^{(\omega)}}_{a_{L_1^{(\omega)}} \vert x_{L_1^{(\omega)}}, \omega} \otimes \sigma^{K_2^{(\omega)}}_{a_{L_2^{(\omega)}} \vert x_{L_2^{(\omega)}}, \omega},
\end{align}
where, for each $\mu$ and $\omega$, $L^{(\cdot)}_1$ and $L^{(\cdot)}_2$ constitute a bipartition of $\{1, \ldots, k\}$, and for each $\nu$ and $\omega$, $K^{(\cdot)}_1$ and $K^{(\cdot)}_2$ form a bipartition of $\{k+1, \ldots, n\}$.
We also require all steering assemblages and the probability tables appearing in the right-hand side of \cref{eq:general_bisep_assemblage} to admit a quantum realisation.

We highlight two special features of this definition, which will be useful for some of our proofs.
First, starting from an assemblage $ \mathcal{S} := \{\sigma^{k+1 \ldots n}_{a_1 \ldots a_k \vert x_1 \ldots x_k}\}$, consider the assemblage $\mathcal{S'} := \{\sigma^{k+1 \ldots n'}_{a_1 \ldots a_k \vert x_1 \ldots x_k}\}$, where the last $n-n'$ parties are merged into a single party labelled by $n'$. %the parties $n-n', \ldots, n$ are merged into a single party $n'$.
If $\mathcal{S'}$ does not exhibit GMS, it follows that $ \mathcal{S}$ also does not exhibit GMS.

Second, starting from an assemblage in a $k$-sided device-independent scenario $\mathcal{S}:=\{\sigma^{k+1 \ldots n}_{a_1 \ldots a_k \vert x_1 \ldots x_k}\}$, we may construct an assemblage in a $k+1$-sided device-independent scenario via $\mathcal{S'} := \{\sigma^{k+2 \ldots n}_{a_1 \ldots a_{k+1} \vert x_1 \ldots x_{k+1}}\}$ by $\sigma^{k+2 \ldots n}_{a_1 \ldots a_{k+1} \vert x_1 \ldots x_{k+1}} = \Tr_{k+1}(M_{a_{k+1} \vert x_{k+1}} \otimes \id^{k + 2 \ldots n} \sigma^{k+1 \ldots n}_{a_1 \ldots a_{k} \vert x_1 \ldots x_{k}})$.
If $ \mathcal{S}$ does not exhibit GMS, it follows that $\mathcal{S'}$ also does not exhibit GMS.

\subsection{Useful lemmas}
In the proof of our main results, we use the following lemmas.
\begin{lemma}\label{lem:psd_orthogonality_decomposition}
    Let $A$ be a positive semidefinite operator that can be decomposed as $A = \sum_{k = 1}^N B_k$, where each $B_k$ is also positive semidefinite.
    If $A\ket{\psi} = 0$, then $B_k\ket{\psi} = 0$ for all $k \in \{1, \ldots, N\}$.
\end{lemma}
\begin{proof}
    If $A\ket{\psi} = 0$, then $\bra{\psi}A\ket{\psi} = 0$ and consequently $\bra{\psi}B_k\ket{\psi} = 0$ for every $k$.
    Defining $\ket{\phi_k} = \sqrt{B_k}\ket{\psi}$, the above relation implies that $\ket{\phi_k} = 0$.
    Therefore, $B_k \ket{\psi} = \sqrt{B_k}\ket{\phi_k} = 0$. 
\end{proof}

\begin{lemma}\label{lem:flip_invariance}
    Let $A$ be a positive semidefinite operator that can be decomposed as $A = \sum_{k = 1}^N B_k$, where each $B_k$ is also positive semidefinite, and let  $U$ be a unitary operator.
    If $AU = A$, then $B_k U = B_k$ for all $k \in \{1, \ldots, N\}$.
\end{lemma}
\begin{proof}
    Since $A U = A$, for any vector $\ket{\psi}$ it holds that ${A(\ket{\psi}-U\ket{\psi})=0}$.
    Thus, using \cref{lem:psd_orthogonality_decomposition} we conclude that $B_k (\ket{\psi}-U\ket{\psi})= 0$, which implies $B_k U = B_k$.
\end{proof}

\begin{lemma}\label{lem:pos_left_inverse}
    Let $\ket{\psi^{(1\ldots n)}} \in \mathbb{C}^{d_1 \ldots d_n}$ be an $n$-partite pure state with Schmidt rank $d_1$ in the bipartition $1 \vert 2\ldots n$.
    Then, the map $\Gamma_{\psi}: L(\mathbb{C}^{d_1}) \rightarrow L(\mathbb{C}^{d_2 \ldots d_n})$ defined by $\Gamma_{\psi}(M) = \Tr_1 (M \otimes \id^{2 \ldots n} \ketbra{\psi^{(1 \ldots n)}})$ has a positive left inverse.
\end{lemma}
\begin{proof}
    Let $\ket{\psi^{(1 \ldots n)}} = \sum_{i = 0}^{d_1 - 1} \sqrt{p_i}\ket{\phi_i^{(1)}} \otimes \ket{\zeta^{(2 \ldots n)}_i}$ be the Schmidt decomposition of the state $\ket{\psi^{(1 \ldots n)}}$ in the bipartition $1 \vert 2 \ldots n$.
    Since it has Schmidt rank $d_1$, we have $p_i > 0$ for all $i \in \{0, \ldots, d_1 -1\}$.
    Then, it is straightforward to check that the map $\Lambda_\psi: L(\mathbb{C}^{d_2 \ldots d_n}) \rightarrow
     L(\mathbb{C}^{d_1})$ defined by $\Lambda_\psi(N) = \sum_{i, j = 0}^{d_1 - 1} \frac{1}{\sqrt{p_i p_j}} \bra{\zeta_i^{(2 \ldots n)}} N \ket{\zeta_j^{(2 \ldots n)}}\ketbra{\phi_j^{(1)}}{\phi_i^{(1)}}$ is a left inverse for $\Gamma_\psi$.
     Moreover, note that we can write $\Lambda_\psi (N) = B N^T B^\dagger$, where $B = \sum_{i = 0}^{d_1 -1}\frac{1}{\sqrt{p_i}} \ketbra{\phi_i^{(1)}}{\zeta_i^{(2 \ldots n)}}$, which implies that $\Lambda_\psi$ is positive.
\end{proof}

The following lemma may be obtained as a particular case of a more general result proven in the appendix of \cite{Ichikawa2008SymmEnt}. For completeness, we present here a proof tailored for our case of interest.
\begin{lemma}\label{lem:product_implies_full_product}
    Let $K_1, K_2 \subset \{1, \ldots, n\}$ be a bipartition of $n$ parties, and let $\rho^{1 \ldots n} \in \mathbb{C}^{d^{\otimes n}}$ be an $n$-partite state which factorises in such a bipartition, i.e., $\rho^{1 \ldots n} = \rho^{K_1} \otimes \rho^{K_2}$.
    If $\rho^{1 \ldots n}$ is invariant under the permutation of any two parties, then $\rho = \sigma^{\otimes n}$, where $\sigma = \Tr_{1 \ldots i-1, i+1, \ldots n} (\rho^{1 \ldots n})$ for any $i \in \{1 \ldots n\}$.
\end{lemma}
\begin{proof}
    Since $\rho^{1 \ldots n}$ is invariant under the permutation of any two parties, we can assume that $\rho^{1 \ldots n} = \rho^{1 \ldots k} \otimes \rho^{k + 1 \ldots n}$ for $k = \vert K_1 \vert$.
    Or, equivalently, $\rho^{1 \ldots n} = \Tr_{k+1 \ldots n}(\rho^{1 \ldots n}) \otimes \Tr_{1 \ldots k}(\rho^{1 \ldots n})$.
    Using the invariance of $\rho^{1 \ldots n}$ under the exchange of parties $k$ and $k+1$, we can also write $\rho^{1 \ldots n} = \rho^{1 \ldots k-1, k + 1} \otimes \rho^{k, k+2 \ldots n} = \Tr_{k, k+2 \ldots n}(\rho^{1 \ldots n}) \otimes \Tr_{1 \ldots k-1, k+1}(\rho^{1 \ldots n})$.
    Therefore, $\rho^{1 \ldots n} = \Tr_{k \ldots n}(\rho^{1 \ldots n}) \otimes \Tr_{1 \ldots k-1, k+1, \ldots n}(\rho^{1 \ldots n}) \otimes \Tr_{1 \ldots k}(\rho^{1 \ldots n}) = \rho^{1 \ldots k-1} \otimes \rho^k \otimes \rho^{k+1 \ldots n}$, and the party $k$ is uncorrelated from the others.
    But since the state is symmetric under any permutation of parties, it follows that $\rho = \sigma^{\otimes n}$.
\end{proof}

\subsection{Proof of \cref{thm:one-sided_incomp_implies_gms}}
\indent In the case of one-sided device-independent scenarios, i.e., scenarios with a single untrusted party, the definition \cref{eq:general_bisep_assemblage} reads
\begin{equation}\label{eq:general_one-sided_bisep_assemblage}
    \sigma^{2\ldots n}_{a \vert x} = \sum_\lambda p(\lambda)p(a \vert x, \lambda) \rho^{2 \ldots n}_\lambda + \sum_\mu p(\mu) \sigma^{K^{(\mu)}_1}_{a \vert x \mu} \otimes \rho^{K^{(\mu)}_2}_\mu,
\end{equation}
where, for each $\mu$, $K^{(\mu)}_1$ and $K^{(\mu)}_2$ form a bipartition of $\{2, \ldots, n\}$.
With such definition, the proof of \cref{thm:one-sided_incomp_implies_gms} for an arbitrary number of parties is analogous to the one discussed in the main text for the tripartite case.

Assume that $\sigma_{a \vert x}^{2 \ldots n} = \Tr_1(A_{a|x}\otimes \id^{2 \ldots n} \ketbra{\psi^{(1 \ldots n)}})$ admits a decomposition as in \cref{eq:general_one-sided_bisep_assemblage}.
Since $\ket{\psi^{(1 \ldots n)}}$ is invariant under any permutation of parties $2, \ldots, n$, it follows that $\sigma_{a \vert x}^{2 \ldots n} V_\pi = \sigma_{a \vert x}^{2 \ldots n}$, where $V_\pi:{\mathbb{C}^d}^{n-1}\to {\mathbb{C}^d}^{n-1}$ is the unitary that performs the $\pi\in S_{n-1}$ permutation on the parties $2 \ldots n$.
\cref{lem:psd_orthogonality_decomposition} and \cref{lem:flip_invariance} can then be used to conclude that $\sigma^{K^{(\mu)}_1}_{a \vert x \mu} \otimes \rho^{K^{(\mu)}_2}_\mu V_\pi = \sigma^{K^{(\mu)}_1}_{a \vert x \mu} \otimes \rho^{K^{(\mu)}_2}_\mu$ for every $\mu$ and any permutation $V_\pi$. Since $\sigma^{K^{(\mu)}_1}_{a \vert x \mu} \otimes \rho^{K^{(\mu)}_2}_\mu={\sigma^{K^{(\mu)}_1}_{a \vert x \mu} \otimes \rho^{K^{(\mu)}_2}_\mu}^\dagger$, we must then have $V_\pi^\dagger \sigma^{K^{(\mu)}_1}_{a \vert x \mu} \otimes \rho^{K^{(\mu)}_2}_\mu V_\pi= \sigma^{K^{(\mu)}_1}_{a \vert x \mu} \otimes \rho^{K^{(\mu)}_2}_\mu$.
That is, the operator $\sigma^{K^{(\mu)}_1}_{a \vert x \mu} \otimes \rho^{K^{(\mu)}_2}$ is invariant under any permutation of the parties $2 \ldots n$.
Now, \cref{lem:product_implies_full_product} leads to $\sigma^{K^{(\mu)}_1}_{a \vert x \mu} \otimes \rho^{K^{(\mu)}_2} = \Tr(\sigma^{K^{(\mu)}_1}_{a \vert x \mu}) \rho_\mu^{\otimes (n-1)}$
This allows us to write $\sigma_{a \vert x}^{2 \ldots n} = \sum_{\lambda'}p(\lambda')p(a \vert x \lambda') \rho_{\lambda'}^{2 \ldots n}$, and using \cref{lem:pos_left_inverse} we observe that the measurement $\mathcal{G} = \{G_{\lambda'}\}$ defined by $G_{\lambda'} = \Lambda_\psi(p(\lambda') \rho_{\lambda'}^{2 \ldots n})$ is a parent measurement for $\{A_{a \vert x}\}_{a,x}$, where $\Lambda_{\psi}$ is a positive left inverse of the map $\Gamma_{\psi}(M) = \Tr_1 (M \otimes \id^{2 \ldots n} \ketbra{\psi^{(1 \ldots n)}})$. \qed

\subsection{Proof of \cref{thm:two_sided_incomp_but_no_gms}}
As discussed in the main text, to prove \cref{thm:two_sided_incomp_but_no_gms} for the tripartite case, we just need to come up with positive semidefinite operators $\{G_{\lambda, b \vert y}\}_{\lambda, b, y}$ and $\{H_{\mu, a \vert x}\}_{\mu, a, x}$ satisfying $\sum_{b}G_{\lambda, b \vert y} = G_\lambda$, $\sum_a H_{\mu, a \vert x} = H_\mu$ such that the measurements $M_{\pm \vert x} = \frac{1}{2}\left(\id \pm \frac{1}{\sqrt[4]{2}}\sigma_x\right)$ can be decomposed as in \cref{eq:bi-partite_jm_no_F}.
One can check that if
\begin{align}
    G_{a_0, a_1,b \vert y} & = \left(\frac{1}{4}\id + \alpha\sum_{x = 0}^1 a_x \sigma_x\right) \otimes \frac{\id}{4} \nonumber\\
    & + \left( \beta \id + \frac{1}{4\sqrt{2}} \sum_{x = 0}^1 a_x \sigma_x\right) \otimes \frac{b\sigma_y}{4}
\end{align}
and
\begin{align}
    H_{b_0, b_1, a \vert x} & = \frac{\id}{4} \otimes \left(\frac{1}{4}\id  + \alpha \sum_{y = 0}^1 b_y \sigma_y\right) \nonumber \\
    & + \frac{a\sigma_x}{4}  \otimes \left(\beta \id  + \frac{1}{4\sqrt{2}} \sum_{y = 0}^1 b_y \sigma_y\right),
\end{align}
where $\alpha = \frac{\sqrt{2}-1}{2\sqrt[4]{2}}$ and $\beta = \frac{1}{2\sqrt[4]{2}} - \alpha $, we have $G_{a_0, a_1, b \vert y} \geq 0$, $H_{b_0, b_1, a \vert x} \geq 0$, $\sum_{b}G_{a_0, a_1, b \vert y} = G_{a_0, a_1}$, $\sum_a H_{b_0, b_1, a \vert x} = H_{b_0, b_1}$ and it also holds that
\begin{equation}
    M_{a \vert x} \otimes M_{b \vert y} = \sum_{a_{x \oplus 1} = \pm 1} G_{a_0, a_1, b \vert y} + \sum_{b_{y \oplus 1} = \pm 1} H_{b_0, b_1, a \vert x}.
\end{equation}
This verification may be done with symbolic computational methods.

Therefore, for any tripartite state $\rho^{ABC}$, the assemblage $\Tr_{AB}(M_{a \vert x} \otimes M_{b \vert y} \otimes \id^C \rho^{ABC})$ does not exhibit genuine multipartite steering.
In fact, if we take any $n$-partite state $\rho^{1 \ldots n}$, the assemblage $\Tr_{1,2}(M_{a \vert x} \otimes M_{b \vert y} \otimes \id^{3 \ldots n} \rho^{1 \ldots n})$ does not exhibit GMS, since if the parties $3, \ldots, n$ are merged into a single party $C$, the above reasoning implies that the assemblage in the corresponding tripartite scenario does not exhibit GMS.
Finally, to conclude the proof of \cref{thm:two_sided_incomp_but_no_gms}, it suffices to notice that performing extra measurements on an assemblage which does not demonstrate GMS also leads to non-GMS assemblages, as discussed in the paragraphs following \cref{eq:general_bisep_assemblage}.
\qed

\subsection{Proof of \cref{thm:bi_partite_jm_iff_no_two-sided_gms}}
First, assume that $\{A_{a \vert x}\}_{a,x}$ and $\{B_{b \vert y}\}_{b,y}$ can be written as \cref{eq:bipartite-jm}.
Then, for any tripartite state $\ket{\psi} \in \mathbb{C}^{d_A} \otimes \mathbb{C}^{d_B} \otimes \mathbb{C}^{d_C}$, it follows that $\Tr_{AB}(A_{a \vert x} \otimes B_{b \vert y} \otimes \id^{C} \ketbra{\psi})$ does not exhibit GMS.
Therefore, this conclusion also holds for any $n$-partite state $\ket{\psi} \in \mathbb{C}^{d_A} \otimes \mathbb{C}^{d_B} \otimes \mathbb{C}^{d_3} \otimes \ldots \otimes \mathbb{C}^{d_n}$, since if we merge the parties $3, \ldots, n$ into a single party $C$, the resulting assemblage is not GMS.

To prove the other direction, let us now assume that for any $n$-partite state $\ket{\psi} \in \mathbb{C}^{d_A} \otimes \mathbb{C}^{d_B} \otimes \mathbb{C}^{d_3} \otimes \ldots \otimes \mathbb{C}^{d_n}$, the assemblage $\Tr_{AB}(A_{a \vert x} \otimes B_{b \vert y} \otimes \id^{3 \ldots n} \ketbra{\psi})$ does not exhibit GMS.
In particular, notice that if we consider a tripartite state $\ket{\psi} = \frac{1}{\sqrt{d_A d_B}}\sum_{i = 0}^{d_A - 1} \sum_{j = 0}^{d_B - 1} \ket{i}_A \otimes \ket{j}_B \otimes \ket{\phi_{ij}}_C$, with $\braket{\phi_{ij}}{\phi_{i' j'}} = \delta_{i i'} \delta_{j j'}$, we can use \cref{lem:pos_left_inverse} together with \cref{eq:two-sided_bisep_assemblage} to construct a decomposition of $\{A_{a \vert x}\}_{a,x}$ and $\{B_{b \vert y}\}_{b,y}$ as \cref{eq:bipartite-jm}. 
\qed

\begin{acronym}[CGLMP]\itemsep 1\baselineskip
\acro{AGF}{average gate fidelity}
\acro{AMA}{associated measurement assemblage}

\acro{BOG}{binned outcome generation}

\acro{CGLMP}{Collins-Gisin-Linden-Massar-Popescu}
\acro{CHSH}{Clauser-Horne-Shimony-Holt}
\acro{CP}{completely positive}
\acro{CPT}{completely positive and trace preserving}
\acro{CPTP}{completely positive and trace preserving}
\acro{CS}{compressed sensing} 

\acro{DFE}{direct fidelity estimation} 
\acro{DM}{dark matter}

\acro{GST}{gate set tomography}
\acro{GPT}{general probabilistic theory}
\acroplural{GPT}[GPTs]{general probabilistic theories}
\acro{GUE}{Gaussian unitary ensemble}

\acro{HOG}{heavy outcome generation}

\acro{JM}{jointly measurable}

\acro{LHS}{local hidden-state model}
\acro{LHV}{local hidden-variable model}
\acro{LOCC}{local operations and classical communication}

\acro{MBL}{many-body localization}
\acro{ML}{machine learning}
\acro{MLE}{maximum likelihood estimation}
\acro{MPO}{matrix product operator}
\acro{MPS}{matrix product state}
\acro{MUB}{mutually unbiased bases} 
\acro{MW}{micro wave}

\acro{NISQ}{noisy and intermediate scale quantum}

\acro{POVM}{positive operator valued measure}
\acro{PR}{Popescu-Rohrlich}
\acro{PVM}{projector-valued measure}

\acro{QAOA}{quantum approximate optimization algorithm}
\acro{QML}{quantum machine learning}
\acro{QMT}{measurement tomography}
\acro{QPT}{quantum process tomography}
\acro{QRT}{quantum resource theory}
\acroplural{QRT}[QRTs]{Quantum resource theories}

\acro{RDM}{reduced density matrix}

\acro{SDP}{semidefinite program}
\acro{SFE}{shadow fidelity estimation}
\acro{SIC}{symmetric, informationally complete}
\acro{SM}{Supplemental Material}
\acro{SPAM}{state preparation and measurement}

\acro{RB}{randomized benchmarking}
\acro{rf}{radio frequency}

\acro{TT}{tensor train}
\acro{TV}{total variation}

\acro{UI}{uninformative}

\acro{VQA}{variational quantum algorithm}

\acro{VQE}{variational quantum eigensolver}

\acro{WMA}{weighted measurement assemblage}

\acro{XEB}{cross-entropy benchmarking}

\end{acronym}

\end{document}